\documentclass[%
 aps, physrev,
]{revtex4-2}

\usepackage{graphicx, amsmath, amssymb, amsfonts, amsthm}
\usepackage{dcolumn}
\usepackage{bm}

\newtheorem{theorem}{Theorem}
\newtheorem{lemma}{Lemma}
\newtheorem{corollary}{Corollary}

\theoremstyle{definition}
\newtheorem{definition}{Definition}

\newtheorem{observation}{Observation}

\newcommand{\outpro}[2]{\vert #1\rangle\langle #2\vert}

\newcommand{\ket}[1]{\vert #1\rangle}
\newcommand{\bra}[1]{\langle #1\vert}

\newcommand{\ptr}[2]{\mathsf{Tr}_{#1}(#2)}

\begin{document}


\title{\textbf{On properties of Schmidt Decomposition} 
}%

\author{Mithilesh Kumar}
 \email{Contact author: admin@drmithileshkumar.com}
 \homepage{https://drmithileshkumar.com}



\begin{abstract}
Schmidt decomposition is a powerful tool in quantum information. While Schmidt decomposition is universal for bipartite states, its not for multipartite states. In this article, we review properties of bipartite Schmidt decompositions and study which of them extend to multipartite states. In particular, Schmidt number (the number of non-zero terms in Schmidt decomposition) define an equivalence class using separable unitary transforms. We show that it is NP-complete to partition a multipartite state that attains the highest Schmidt number. In addition, we observe that purifications of a density matrix of a composite system preserves Schmidt decomposability.
\end{abstract}

\maketitle


\section{\label{sec:intro}Introduction}
Composite quantum systems exhibit a special property called entanglement. Entangled quantum subsystems remain connected independent of separation between them, we don't fully understand how. Entanglement is now considered a physical resource. We can \emph{spend} entanglement within a system to accomplice tasks like quantum communication, teleportation, superdense coding etc. Naturally, we would like to measure amount of entanglement in a state. One method is to rewrite the state in a specific form called Schmidt decomposition and the number of non-zero terms in the decomposition quantify the amount of entanglement. In the simplest case, a quantum system is composed of two subsystems \(A\) and \(B\) whose state in standard orthonormal basis can be expressed as 
\begin{align}\label{eq:bipartite}
    \ket{\psi} = \sum_{i,j} a_{ij}\ket{i_A}\ket{i_B}
\end{align}
where \(a_{ij}\) are complex numbers. The state vectors \(\ket{i_A}\) and \(\ket{i_B}\) form orthonormal basis of state spaces \(\mathcal{H}_A\) and \(\mathcal{H}_B\) respectively. The indices \(i, j\) run over the dimensions of each of the spaces. A restatement of Singular Value Decomposition implies that the bipartite state \ref{eq:bipartite} can always be rewritten in another bases of \(\mathcal{H}_A\) and \(\mathcal{H}_B\) as
\begin{align}\label{eq:bischmidt}
    \ket{\psi} = \sum_{\ell} \lambda_\ell\ket{\ell_A}\ket{\ell_B}
\end{align}
where \(\lambda_\ell\) can be chosen to be real and non-negative. The states \(\ket{\ell_A}\) and \(\ket{\ell_B}\) are orthonormal within their respective spaces (but need not form a basis) \cite{BookPeres}. The number of non-zero terms in the state \ref{eq:bischmidt} is called the \emph{Schmidt number} of the state. If the Schmidt number of a state is \(1\), i.e. there is only one term in the decomposition, then the system does not exhibit any entanglement. As the state \(\ket{\psi}\) is called a \emph{product state}. Higher is the Schmidt number, larger is the amount of entanglement within the system. Terhal and Horodecki \cite{HorodeckiNumber} extended the notion of Schmidt number to bipartite density matrices, but in this article, we restrict our inquiry to Schmidt number of pure states only. Vidal \cite{Vidal} and Nest \cite{Nest} showed that states with high Schmidt number are necessary for quantum speed-up.

Naturally, there exist entanglement within a system even if it can been seen as composed of more than two subsystems. Such systems are called multipartite. We can label each part as \(A_1, A_2, ..., A_n\) and a multipartite state can be expanded using orthonormal basis states as 
\begin{align}\label{eq:multipartite}
    \ket{\psi} = \sum_{i_1,i_2, ..., i_n} a_{i_1i_2...i_n}\ket{i^{A_1}_1}\ket{i^{A_2}_2}\cdots\ket{i^{A_n}_n}
\end{align}
where \(a_{i_1i_2...i_n}\) are complex numbers satisfying normalization condition \(\sum |a_{i_1i_2...i_n}|^2 = 1\). The states \(\ket{i^{A_1}_1},\ket{i^{A_2}_2},\cdots\), \(\ket{i^{A_n}_n}\) are orthonormal within their respective spaces \(\mathcal{H}_{A_k}\). The Schmidt decomposition is just an extension of the bipartite case written as 
\begin{align}\label{eq:multischmidt}
    \ket{\psi} = \sum_{\ell} \lambda_\ell\ket{\ell_{A_1}}\ket{\ell_{A_2}}\cdots\ket{\ell_{A_n}}
\end{align}
where \(\lambda_\ell\) can be chosen to be real and non-negative. The states \(\ket{\ell_{A_k}}\) are orthonormal within their respective spaces (but need not form a basis). It turns out that this form is too strict to be satisfied by all multipartite states. There are multipartite states do not admit Schmidt decomposition. Because of this, Huber and de Vicente \cite{Huber} tried to consider a vector of Schmidt numbers of two dimensional subsystems in tripartite case.  In this work, we'll only consider Schmidt decomposable states.

In this article, we'll review some well known properties of Schmidt number of bipartite states. We'll analyze which of them hold for multipartite states. We'll also consider some of the applications of Schmidt decomposition in quantum information theory.

\section{Classification based on Schmidt bases}\label{sec:schmidtBases}
As has been noted, all bipartite states have Schmidt decomposition. The case for bipartite states is well known and the reader is referred to \cite{nielsen00}
\begin{theorem}
    If \(\ket{\psi}\) is multipartite Schmidt decomposable state, then the eigenvalues of reduced density matrices of any subsystem are the same.
\end{theorem}
\begin{proof}
    Using the Schmidt decomposition
    \begin{align}\label{eq:multischmidt}
    \ket{\psi} = \sum_{\ell} \lambda_\ell\ket{\ell_{A_1}}\ket{\ell_{A_2}}\cdots\ket{\ell_{A_n}}
\end{align}
The density matrix is 
\begin{align}
    \rho = \outpro{\psi}{\psi} = \sum_{\ell, m} \lambda_\ell\lambda_m\ket{\ell_{A_1}}\ket{\ell_{A_2}}\cdots\ket{\ell_{A_n}}\bra{\ell_{A_1}}\bra{\ell_{A_2}}\cdots\bra{\ell_{A_n}}
\end{align}
Tracing out subsystems of any subset of indices \(S\subset \{1, 2, ..., n\}\), we get
    \begin{align}
        \rho^{A_{S^C}} = \sum_\ell \lambda_\ell^2\outpro{\ell_{S^C}}{\ell_{S^C}}
    \end{align}
    where \(S^C = \{1,2,...,n\} - S\) and the states \(\ell_{S^C}\) is product states running on indices over \(S^C\). Clearly all reduced density matrices \(\rho^{A_{S^C}}\) have the same eigenvalues \(\lambda_\ell\).
\end{proof}
Unfortunately, the above condition for Schmidt decomposability is not sufficient as shown by the following example. Consider the following tripartite state
\begin{align}
    \ket{W} = \frac{1}{\sqrt{3}}\left ( \ket{001} + \ket{010} + \ket{100}\right )
\end{align}
It is well known that \(\ket{W}\) does not admit Schmidt decomposition. We can show it using the criteria developed in Kumar \cite{kumarMultipartite}.
\begin{theorem}\label{thm:tripartite}\cite{kumarMultipartite}
    A tripartite state \(\ket{\psi}\) is Schmidt decomposable if and only if
    \begin{enumerate}
        \item the matrix set \(\mathcal{A}\) commutes positively, and
        \item the matrix \(S = [diag(P^\dagger A_i Q^\dagger)]\) is scaled unitary.
    \end{enumerate}
\end{theorem}
The matrix set of \(\ket{W}\) is 
\begin{align*}
    \mathcal{A} = \left\{A_0 = \frac{1}{\sqrt{3}}\begin{bmatrix}
        0 & 1\\
        1 & 0
    \end{bmatrix}, A_1 = \frac{1}{\sqrt{3}}\begin{bmatrix}
        1 & 0\\
        0 & 0
    \end{bmatrix}\right\}
\end{align*}
Computing the positive matrices \(A_0A_0^\dagger, A_0^\dagger A_0, A_1A_1^\dagger, A_1^\dagger A_1\), we get
\begin{align*}
    C_0 &= A_0A_0^\dagger = A_0^\dagger A_0 = \frac{1}{3}\begin{bmatrix}
        1 & 0\\
        0 & 1
    \end{bmatrix}\\
    C_1 &= A_1A_1^\dagger = A_1^\dagger A_1 = \frac{1}{3}\begin{bmatrix}
        1 & 0\\
        0 & 0
    \end{bmatrix}\\
\end{align*}
The matrix set \(\mathcal{A}\) is positively commuting as \(C_0C_1 = C_1C_0\). The diagonalization pair of unitary matrices are
\begin{align*}
    \left\{P = \begin{bmatrix}
        0 & 1\\
        1 & 0
    \end{bmatrix}, Q = \begin{bmatrix}
        0 & 1\\
        1 & 0
    \end{bmatrix}\right\}
\end{align*}
The columns of the matrix \(S\) are diagonals of 
\begin{align*}
    \left\{P^\dagger A_0 Q^\dagger = \frac{1}{\sqrt{3}}\begin{bmatrix}
        1 & 0\\
        0 & 1
    \end{bmatrix}, P^\dagger A_1 Q^\dagger = \frac{1}{\sqrt{3}}\begin{bmatrix}
        0 & 0\\
        0 & 1
    \end{bmatrix}\right\}
\end{align*}
That is, 
\begin{align*}
    S &= \frac{1}{\sqrt{3}} \begin{bmatrix}
        1 & 0\\
        1 & 1
    \end{bmatrix}\\
    SS^\dagger &= \frac{1}{3}\begin{bmatrix}
        1 & 1\\
        1 & 2
    \end{bmatrix} \neq \textbf{diagonal}
\end{align*}
Hence \(\ket{W}\) does not have Schmidt decomposition. 

Now we'll show that reduced density matrices are identical for \(\rho = \outpro{\psi}{\psi}\).
\begin{align*}
    \rho = \outpro{\psi}{\psi} = \frac{1}{3} [&\outpro{001}{001} + \outpro{010}{001} + \outpro{100}{001} + \\
    & \outpro{001}{010} + \outpro{010}{010} + \outpro{100}{010} + \\
    & \outpro{001}{100} + \outpro{010}{100} + \outpro{100}{100}
    ]
\end{align*}
Tracing out two subsystems gives
\begin{align*}
    \rho^A &= \frac{2}{3}\outpro{0}{0} + \frac{1}{3}\outpro{1}{1}\\
    \rho^B &= \frac{2}{3}\outpro{0}{0} + \frac{1}{3}\outpro{1}{1}\\
    \rho^C &= \frac{2}{3}\outpro{0}{0} + \frac{1}{3}\outpro{1}{1}
\end{align*}
As is evident, \(\rho^A, \rho^B, \rho^C\) have same eigenvalues, but \(\ket{\psi}\) is not Schmidt decomposable.

As is noted in \cite{nielsen00}, for a bipartite state, if \(\ket{\psi} = \sum_i\lambda_i\ket{i_A}\ket{i_B}\) is the Schmidt decomposition, then \(\sum_i\lambda_i(U\ket{i_A})(V\ket{i_B})\) are Schmidt decomposition of \((U\otimes V)\ket{\psi}\). This observation extends to multipartite systems as well. Moreover, if two states \(\psi\) and \(\phi\) have the same Schmidt coefficients, then they must be unitary transforms of each other, i.e. \(\ket{\psi} = (U\otimes V)\ket{\phi}\).
\begin{theorem}\label{thm:multiUnitary}
    Suppose \(\ket{\psi}\) and \(\ket{\phi}\) are two Schmidt decomposable multipartite states on the same Hilbert space. Then, \(\ket{\psi} = (U_1\otimes U_2\otimes \cdots U_n)\ket{\phi}\) if and only if the Schmidt decompositions of \(\ket{\psi}\) and \(\ket{\phi}\) have the same Schmidt coefficients.
\end{theorem}
\begin{proof}
    If \(\ket{\psi} = (U_1\otimes U_2\otimes \cdots U_n)\ket{\phi}\), we can use the Schmidt decomposition of \(\ket{\phi}\) as 
    \begin{align*}
        \ket{\phi} &= \sum_\ell \lambda_\ell\ket{\ell_{A_1}}\ket{\ell_{A_2}}\cdots\ket{\ell_{A_n}}\\
        \ket{\psi} &= \sum_\ell \lambda_\ell(U_{A_1}\ket{\ell_{A_1}})(U_{A_2}\ket{\ell_{A_2}})\cdots(U_{A_n}\ket{\ell_{A_n}})\\
         &= \sum_\ell \lambda_\ell\ket{\ell^\psi_{A_1}}\ket{\ell^\psi_{A_2}}\cdots\ket{\ell^\psi_{A_n}}
    \end{align*}
    As states \(\ket{\ell^\psi_{A_k}}\) are obtained via unitary transform of orthonormal states \(\ket{\ell_{A_k}}\), they remain orthonormal and hence provide a valid Schmidt decomposition of \(\ket{\psi}\). Therefore, both \(\psi\) and \(\phi\) have the same Schmidt coefficients.

    For converse, assume \(\psi\) and \(\phi\) have the same Schmidt coefficients, i.e.
    \begin{align*}
        \ket{\phi} &= \sum_\ell \lambda_\ell\ket{\ell_{A_1}}\ket{\ell_{A_2}}\cdots\ket{\ell_{A_n}}\\
        \ket{\psi} &= \sum_\ell \lambda_\ell\ket{\ell^\psi_{A_1}}\ket{\ell^\psi_{A_2}}\cdots\ket{\ell^\psi_{A_n}}
    \end{align*}
    Since states \(\{\ket{\ell_{A_k}}\}\) and \(\{\ket{\ell^\psi_{A_k}}\}\) are orthonormal vectors, there exists unitary matrix \(U_{A_k}\) such that \(\ket{\ell^\psi_{A_k}} = U_{A_k}\ket{\ell_{A_k}}\). This implies that 
    \begin{align*}
        \ket{\psi} &= \sum_\ell \lambda_\ell\ket{\ell^\psi_{A_1}}\ket{\ell^\psi_{A_2}}\cdots\ket{\ell^\psi_{A_n}}\\
        &=\sum_\ell \lambda_\ell(U_{A_1}\ket{\ell_{A_1}})(U_{A_2}\ket{\ell_{A_2}})\cdots(U_{A_n}\ket{\ell_{A_n}})\\
        &= (U_1\otimes U_2\otimes \cdots U_n) \sum_\ell \lambda_\ell\ket{\ell_{A_1}}\ket{\ell_{A_2}}\cdots\ket{\ell_{A_n}}\\
        &= (U_1\otimes U_2\otimes \cdots U_n)\ket{\phi}
    \end{align*}
\end{proof}
Theorem \ref{thm:multiUnitary} implies that if we define Schmidt decomposable states \(\ket{\psi}\) and \(\ket{\phi}\) \emph{equivalent} if \(\ket{\psi} = (U_1\otimes U_2\otimes \cdots U_n)\ket{\phi}\), then the Schmidt number is decided by the number of separable orthonormal basis states. We can consider Schmidt decomposition as establishing a bijection between orthonormal set of vectors of each component space. 
\begin{figure}
    \centering
    \includegraphics[width=0.3\linewidth]{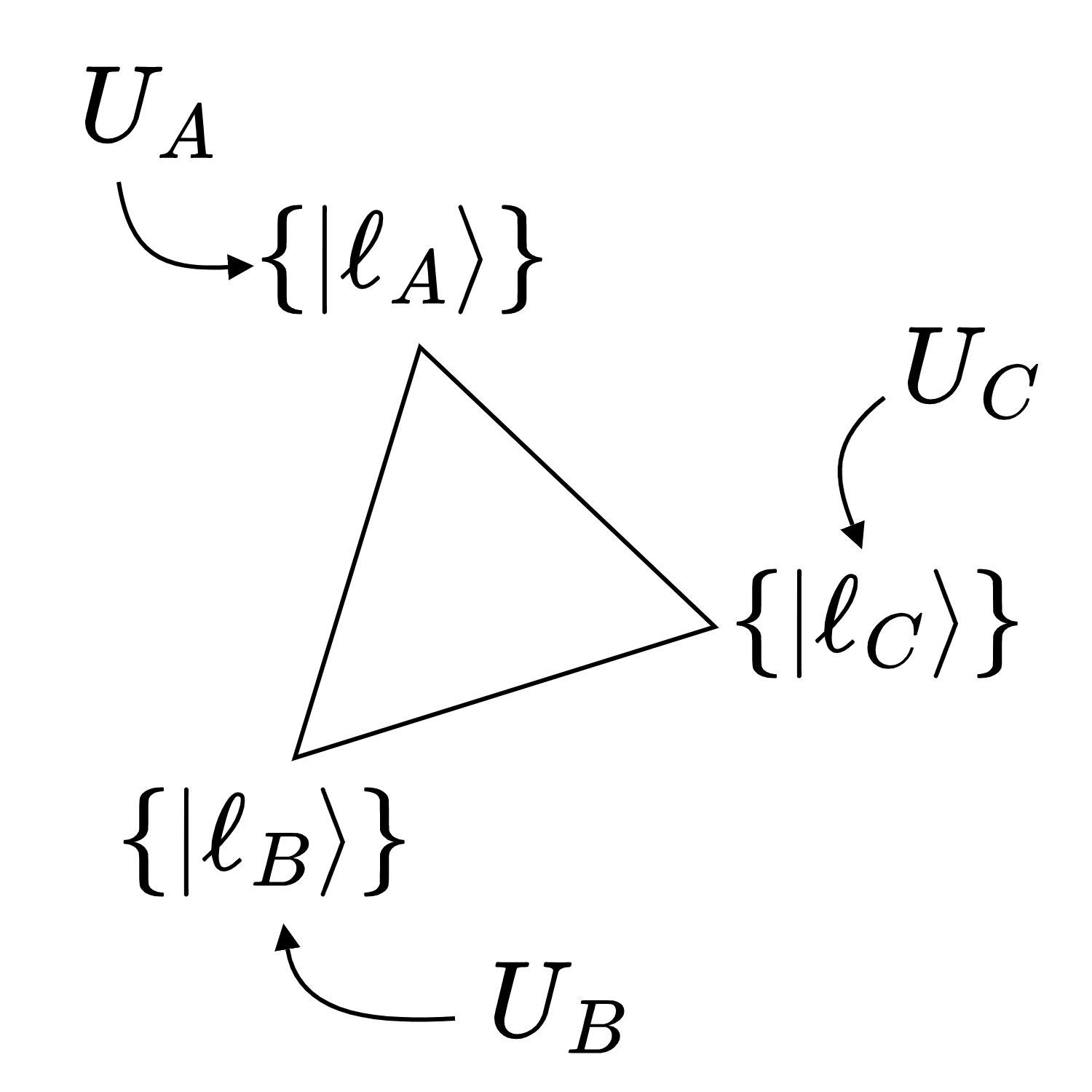}
    \caption{Schmidt decomposition establishes a bijection between orthonormal set of vectors of each subsystem. Equivalent states having the same Schmidt coefficients are generated using unitary operations on individual subsystems.}
    \label{fig:enter-label}
\end{figure}

\begin{observation}
    All two qubit Schmidt bases with Schmidt number \(1\) are equivalent to \(\ket{00}\) and Schmidt number \(2\) are equivalent to \(\{\ket{00},\ket{11}\}\).
\end{observation}
In case of three qubit system, we can partition the system in two ways, (a) one subsystem with \(1\) qubit and another with \(2\) qubits, and (b) three subsystems with \(1\) qubit each. In both the cases, Schmidt rank \(1\) bases are \(\ket{000}\). Schmidt rank \(2\) bases are \(\{\ket{000},\ket{111}\}\). 
\begin{observation}
    There are no Schmidt rank \(3\) states on a three qubit system (irrespective of partition).
\end{observation}
\begin{observation}
    The Schmidt number of an \(n\) qubit system is bounded by \(2^{\lfloor \frac{n}{2}\rfloor}\) irrespective of partition. In fact, the highest Schmidt number is achieved for equal bipartition of the system.
\end{observation}
Let us define the problem SCHMIDT-PARTITION  as follows:\\
\textbf{Input:} \(n\) systems of dimensions \(d_1, d_2,...,d_n\)\\
\textbf{Task:} partition the system into two parts such that there exists a state of Schmidt number \(K\)

\begin{theorem}\label{thm:Spart}
    SCHMIDT-PARTITION is \texttt{NP}-complete.
\end{theorem}
\begin{proof}
    To have a state of Schmidt number \(K\), the dimension of the smallest subsystem must be at least \(K\). The dimension of combined system of two subsystems of dimensions \(d_1\) and \(d_2\) is \(d_1d_2\). Since all the dimensions are positive numbers, we can take the \(\log\) of the dimensions and ask wether there exists bipartition of the system such that the sum of \(\log(d_i)\) for each part is at least \(\log(K)\). This is equivalent to the \texttt{NP}-complete SUBSET SUM problem in which the input is a set of positive numbers and the task is to find a subset whose sum is \(K\).
\end{proof}
The Theorem \ref{thm:Spart} immediately implies the following corollary.
\begin{corollary}
    Given \(n\) systems of dimensions \(d_1, d_2,...,d_n\), it is \texttt{NP}-complete to find a Schmidt decomposable state of highest Schmidt number.
\end{corollary}
\begin{theorem}\label{thm:reduced}
    If \(\ket{\psi} = \sum_{ij}a_{ij}\ket{i_A}\ket{j_B}\) is a bipartite state where \(\ket{i_A}\) and \(\ket{j_B}\) for orthonormal bases of their respective spaces, then the reduced density matrix of any subsystem is given by 
    \begin{align}
        \rho^A &= AA^\dagger\\
        \rho^B &= BB^\dagger = (A^\dagger A)^T
    \end{align}
    where \(A\) is the matrix of coefficients \(a_{ij}\) and \(B = A^T\).
\end{theorem}
\begin{proof}
The density matrix of the system is 
    \begin{align*}
        \rho &= \outpro{\psi}{\psi}\\
        &= \sum_{ijkl}a_{ij}a^*_{kl}\outpro{i_Aj_B}{k_Al_B}
    \end{align*}
Tracing out B gives
\begin{align*}
    \rho_A &= \sum_{ijk}a_{ij}a^*_{kj}\outpro{i_A}{k_A}\\
    &= \sum_{ik}\left(\sum_j a_{ij}a^*_{kj}\right )\outpro{i_A}{k_A}\\
    &= AA^\dagger
\end{align*}
Tracing out A gives
\begin{align*}
    \rho_B &= \sum_{ijl}a_{ij}a^*_{il}\outpro{j_B}{l_B}\\
    &= \sum_{jl}\left(\sum_i a^*_{il}a_{ij}\right )\outpro{j_B}{l_B}\\
    &= A^T (A^T)^\dagger = BB^\dagger
\end{align*}
\end{proof}
Theorem \ref{thm:reduced} can be extended to multipartite states as well. We show this in case of tripartite states. Similar pattern can be used for higher partite states.
\begin{theorem}\label{thm:reducedTripartite}
    If \(\ket{\psi} = \sum_{ijk}a_{ijk}\ket{i_A}\ket{j_B}\ket{k_C}\) is a tripartite state where \(\ket{i_A}, \ket{j_B}\) and \(\ket{k_C}\) for orthonormal bases of their respective spaces, then the reduced density matrices of one-level and two-level subsystems are given by 
    \begin{align}
        \rho^A &= AA^\dagger\\
        \rho^{AB} &= BB^\dagger
    \end{align}
    where \(B\) is the matrix of coefficients \(a_{ij,k}\) of dimension \(n_An_B\times n_C\) and \(A\) is the matrix of coefficients \(a_{i, jk}\) of dimension \(n_A\times n_Bn_C\).
\end{theorem}
\begin{proof}
    Let us calculate \(\rho_A\) first by tracing over \(B\) and \(C\) in the matrix
    \begin{align}\label{eq:tridensity}
        \outpro{\psi}{\psi} = \sum_{ijk, lmn} a_{ijk}a^*_{lmn}\outpro{i_Aj_Bk_C}{l_Am_Bn_C}
    \end{align}
    Tracing out subsystems \(B\) and \(C\) gives
    \begin{align*}
        \rho_A &= \sum_{ijkl}a_{ijk}a^*{ljk}\outpro{i_A}{l_A}\\
               &= \sum_{il}\left(\sum_{jk}a_{ijk}a^*_{ljk}\right )\outpro{i_A}{l_A}\\
               &= \sum_{il}\left(\sum_{jk}a_{i,jk}a^\dagger_{jk,l}\right )\outpro{i_A}{l_A}
    \end{align*}
    In the inner sum, we treat the indices \(jk\) as one giving us a matrix \(A\) of dimension \(n_A\times n_Bn_C\). Therefore,
    \begin{align*}
        \rho_A = AA^\dagger
    \end{align*}
    Next, tracing out \(C\) only gives
    \begin{align*}
        \rho_{AB} &= \sum_{ijk,lm}a_{ijk}a^*_{lmk}\outpro{i_Aj_B}{l_Am_B}\\
                  &= \sum_{ij,lm}\left (\sum_k a_{ij,k}a^*_{lm,k}\right )\outpro{i_Aj_B}{l_Am_B}
    \end{align*}
    In the inner sum, we define matrix \(B\) of elements \(a_{ij, k}\) of dimension \(n_An_B\times n_C\). Therefore,
    \begin{align*}
        \rho_{AB} = BB^\dagger
    \end{align*}
    Reduced density matrices of other subsystems can be obtained by relabeling the indices appropriately.
\end{proof}
It is easy to see that the Schmidt number of a bipartite state is equal to the rank of the reduced density matrix of one of the subsystems by using \(rank(A) = rank(AA^\dagger) = rank(A^\dagger A) = rank((A^\dagger A)^T)\). This result holds for multipartite states as well. 
\begin{theorem}
    If \(\ket{\psi}\) is a Schmidt decomposable multipartite state, then the Schmidt number is equal to the rank of the reduced density matrix of one of the subsystems.
\end{theorem}
\begin{proof}
We show here for tripartite states.
    Consider the Schmidt decomposition of \(\ket{\psi} = \sum_\ell\lambda_\ell\ket{\ell_A}\ket{\ell_B}\ket{\ell_C}\).
    \begin{align*}
        \rho^{ABC} &= \outpro{\psi}{\psi}\\
                   &= \sum_{\ell,m}\lambda_\ell\lambda_m\outpro{\ell_A}{m_A}\otimes\outpro{\ell_B}{m_b}\otimes\outpro{\ell_C}{m_C}\\
        \rho^{AB} &= \sum_\ell^d\lambda_\ell^2\outpro{\ell_A\ell_B}{\ell_A\ell_B}
    \end{align*}
    Consider the spectral decomposition of \(\rho^{AB}\) given as 
    \begin{align*}
        \rho^{AB} = \sum_i^r p_i\outpro{i_{AB}}{i_{AB}}
    \end{align*}
    Since the states \(\{\ket{\ell_A\ell_B}\}\) and \(\{\ket{i_{AB}}\}\) are orthonormal sets, they can each be extended to orthonormal bases using Gram-Schmidt orthogonalization process. This implies that there is unitary matrix \(U\) that takes set of vectors \(\{\ket{\ell_A\ell_B}\}\) to \(\{\ket{i_{AB}}\}\).
    \begin{align*}
        U\rho^{AB}U^\dagger = \sum_i^r p_i\outpro{\ell_A\ell_B}{\ell_A\ell_B}
    \end{align*}
    Since \(\rho^{AB}\) and \(U\rho^{AB}U^\dagger\) are similar matrices, they must have the same rank, i.e. \(d = r\).
\end{proof}
In the book \cite{nielsen00} (attributed to Thapliyal), for a bipartite state \(\ket{\psi} = \alpha \ket{\phi} + \beta\ket{\gamma}\), the Schmidt number satisfy the inequality \(Sch(\psi)\geq |Sch(\phi) - Sch(\gamma)|\). Tt is a simple consequence of Theorem \ref{thm:reduced}. We show that it is true for multipartite Schmidt decomposable states. 
\begin{theorem}
    Given three multipartite Schmidt decomposable states \(\ket{\psi}, \ket{\phi}\) and \(\ket{\gamma}\) such that \(\ket{\psi} = \alpha \ket{\phi} + \beta\ket{\gamma}\), then 
    \begin{align}
        Sch(\psi)\geq |Sch(\phi) - Sch(\gamma)|
    \end{align}
\end{theorem}
\begin{proof}
    We show it for tripartite state for readability.
    \begin{align*}
    \ket{\psi} = \alpha\ket{\phi} + \beta\ket{\gamma} = \sum_{ijk}(\alpha a^\phi_{ijk} + \beta a^\gamma_{ijk})\ket{i_A}\ket{j_B}\ket{k_C} = \sum_{ijk}a_{ijk}\ket{i_A}\ket{j_B}\ket{k_C}
\end{align*}
Using Theorem \ref{thm:reducedTripartite}, we get 
\begin{align*}
    \rho_A &= AA^\dagger\\
           &= (\alpha A_\phi + \beta A_{\gamma})(\alpha A_\phi + \beta A_{\gamma})^\dagger
\end{align*}
This implies that
\begin{align*}
    Sch(\psi) &= rank(\rho_A)\\
              &= rank((\alpha A_\phi + \beta A_{\gamma})(\alpha A_\phi + \beta A_{\gamma})^\dagger)\\
              &= rank(\alpha A_\phi + \beta A_{\gamma})\\
              &\geq |rank(\alpha A_\phi) - rank(\beta A_{\gamma})|\\
              &= |rank(A_\phi) - rank(A_{\gamma}) |\\
              &= |Sch(\phi) - Sch(\gamma)|
\end{align*}
\end{proof}

\begin{observation}
    Linear combination of Schmidt decomposable multipartite states is not necessarily Schmidt decomposable. An an example, we can work with \(\ket{W} = \frac{1}{\sqrt{3}}(\ket{001}) + \ket{010} + \ket{100}\) which is not Schmidt decomposable, but each of the terms are Schmidt decomposable.
\end{observation}
In light of the above observation, we prove the following result for tensor product of Schmidt decomposable states.
\begin{definition}
    An \(n\)-partite Schmidt decomposable state of Schmidt rank \(k\) is said to have Schmidt dimension \((n,k)\).
\end{definition}
\begin{theorem}
    Given two states \(\ket{\psi}\) and \(\ket{\phi}\) of Schmidt dimensions \((m,k_\psi)\) and \((n, k_\phi)\) such that \(m\geq n\), then \(\ket{\psi}\ket{\phi}\) has dimension \((n, k_\psi k_\phi)\) and there are \({m-1\choose n-1}\) ways to construct the \(n\)-partite systems.
\end{theorem}
\begin{proof}
    Starting with the Schmidt decompositions of \(\ket{\psi}\) and \(\ket{\phi}\)
    \begin{align*}
        \ket{\psi} &= \sum_{i}^{k_\psi} \lambda_i^\psi \ket{i_{A_1}...i_{A_m}}\\
        \ket{\phi} &= \sum_{j}^{k_\phi} \lambda_j^\phi \ket{j_{B_1}...i_{B_n}}\\
        \ket{\psi}\ket{\phi} &= \sum_{i,j}^{k_\psi, k_\phi} \lambda_i^\psi \lambda_j^\phi \ket{i_{A_1}...i_{A_m}j_{B_1}...j_{B_n}}
    \end{align*}
    Notice that the above expression for \(\ket{\psi}\ket{\phi}\) is not a Schmidt decomposition of \((m + n)\)-partite system as we lose the bijection between states of each part. We partition \(m\)-parts of \(\ket{\psi}\) into \(n\) groups such that each group has at least 1 part. Let \(x_1, ..., x_n\) be the sizes of each groups. For each group, we link one part from \(\ket{\phi}\) as shown below:
    \begin{align*}
        \ket{\psi}\ket{\phi} &= \sum_{i,j}^{k_\psi, k_\phi} \lambda_i^\psi \lambda_j^\phi \ket{i_{A_1}...i_{A_{x_1}}j_{B_1}}\ket{i_{A_{x_1 + 1}}...A_{x_1 + x_2}j_{B_2}}...\ket{i_{m-(x_n) + 1}...i_{m}j_{B_n}}
    \end{align*}
    It is easy to verify that states of each group are orthonormal and they don't repeat in the sum, giving a valid Schmidt decomposition. Hence, \(\ket{\psi}\ket{\phi}\) can be partitioned into \(n\) parts such that the \(n\)-partite system has Schmidt decomposition of rank \(k_\psi k_\phi\).

    The number of such groups for fixed permutation of parts in \(\ket{\psi}\) and \(\ket{\phi}\) is equal to the number of solutions of 
    \begin{align*}
        x_1 + x_2 + \cdots + x_n &= m\\
        x_i &\geq 1
    \end{align*}
    which is given by \({m-1\choose n-1}\).
\end{proof}
\section{Purification}
Given a density matrix \(\rho\) of a system \(A\), the task of purification is to attach another subsystem \(R\) such that the combine system is in pure state \(\ket{AR}\) and the tracing out \(R\) gives back \(\rho\). The standard trick \cite{Schrödinger_1936, HUGHSTON199314, Jaynes, Hadjisavvas, Gisin} is to keep the dimension of \(R\) to be at least the rank of \(\rho\) and an embed orthonormal basis into \(\rho\) as follows starting with the spectral decomposition of \(\rho\).
\begin{align*}
    \rho &= \sum_i \lambda_i \outpro{i_A}{i_A}\\
    \ket{AR} &= \sum_i \sqrt{\lambda_i}\ket{i_A}\ket{i_R}\\
    \rho &= \ptr{R}{\outpro{Ar}{Ar}}\\
         &= \sum_i \lambda_i \outpro{i_A}{i_A}
\end{align*}
Essentially, we have used the Schmidt decomposition of \(\ket{AR}\).

Regarding purification, one can ask some interesting questions like given a density matrix of a composite system, can it be purified to a Schmidt decomposable state? If a multipartite state is not Schmidt decomposable, can an additional system make it Schmidt decomposable?
\begin{lemma}\cite{nielsen00}\label{thm:twopure}
    Two purifications \(\ket{AR_1}\) and \(\ket{AR_2}\) of a state \(\rho\) are unitarily linked as \begin{align*}
        \ket{AR_1} = (I\otimes U_R)\ket{AR_2}
    \end{align*} 
\end{lemma}
\begin{proof}
    The key idea is that two orthonormal sets are unitarily connected. Using the spectral decomposition of \(\rho\), we have \(\rho = \sum_i \lambda_i\outpro{i_A}{i_A}\).
    \begin{align*}
        \ket{AR_1} &= \sum_i\sqrt{\lambda_i}\ket{i_A}\ket{i^R_1}\\
        \ket{AR_2} &= \sum_i\sqrt{\lambda_i}\ket{i_A}\ket{i^R_2}\\
    \end{align*}
    Since the orthonormal sets \(\{\ket{i^R_1}\}\) and \({\ket{i^R_2}}\) can be connected by unitary matrix \(U\) such that \(\ket{i^R_1} = U\ket{i^R_2}\), we have
    \begin{align*}
        \ket{AR_1} &= \sum_i\sqrt{\lambda_i}\ket{i_A}\ket{i^R_1}\\
                   &= \sum_i\sqrt{\lambda_i}\ket{i_A}U\ket{i^R_2}\\
                   &= (I\otimes U) \sum_i\sqrt{\lambda_i}\ket{i_A}\ket{i^R_2}\\
                   &= (I\otimes U)\ket{AR_2}
    \end{align*}
\end{proof}
\begin{theorem}
    Every purification of a multipartite system \(\rho\) is either Schmidt decomposable or not Schmidt decomposable.
\end{theorem}
\begin{proof}
    The proof follows by combining Lemma \ref{thm:twopure} with Theorem \ref{thm:multiUnitary}.
\end{proof}
\section{Conclusion}
Schmidt decomposition remains to be an active area of research in quantum computation and quantum information. In this work, we have reviewed some of the known properties of Schmidt decomposition with respect to state vector. We saw that finding partition of highest Schmidt number is NP-complete and purification of a multipartite density matrix falls in one of the two categories, either it has Schmidt decomposition or not. Future work will involve other extensions and generalization of Schmidt number, particularly of mixed states.
\nocite{*}

\bibliography{main}

\begin{thebibliography}{12}%
\makeatletter
\providecommand \@ifxundefined [1]{%
 \@ifx{#1\undefined}
}%
\providecommand \@ifnum [1]{%
 \ifnum #1\expandafter \@firstoftwo
 \else \expandafter \@secondoftwo
 \fi
}%
\providecommand \@ifx [1]{%
 \ifx #1\expandafter \@firstoftwo
 \else \expandafter \@secondoftwo
 \fi
}%
\providecommand \natexlab [1]{#1}%
\providecommand \enquote  [1]{``#1''}%
\providecommand \bibnamefont  [1]{#1}%
\providecommand \bibfnamefont [1]{#1}%
\providecommand \citenamefont [1]{#1}%
\providecommand \href@noop [0]{\@secondoftwo}%
\providecommand \href [0]{\begingroup \@sanitize@url \@href}%
\providecommand \@href[1]{\@@startlink{#1}\@@href}%
\providecommand \@@href[1]{\endgroup#1\@@endlink}%
\providecommand \@sanitize@url [0]{\catcode `\\12\catcode `\$12\catcode `\&12\catcode `\#12\catcode `\^12\catcode `\_12\catcode `\%12\relax}%
\providecommand \@@startlink[1]{}%
\providecommand \@@endlink[0]{}%
\providecommand \url  [0]{\begingroup\@sanitize@url \@url }%
\providecommand \@url [1]{\endgroup\@href {#1}{\urlprefix }}%
\providecommand \urlprefix  [0]{URL }%
\providecommand \Eprint [0]{\href }%
\providecommand \doibase [0]{https://doi.org/}%
\providecommand \selectlanguage [0]{\@gobble}%
\providecommand \bibinfo  [0]{\@secondoftwo}%
\providecommand \bibfield  [0]{\@secondoftwo}%
\providecommand \translation [1]{[#1]}%
\providecommand \BibitemOpen [0]{}%
\providecommand \bibitemStop [0]{}%
\providecommand \bibitemNoStop [0]{.\EOS\space}%
\providecommand \EOS [0]{\spacefactor3000\relax}%
\providecommand \BibitemShut  [1]{\csname bibitem#1\endcsname}%
\let\auto@bib@innerbib\@empty
\bibitem [{\citenamefont {Peres}(1997)}]{BookPeres}%
  \BibitemOpen
  \bibfield  {author} {\bibinfo {author} {\bibfnamefont {A.}~\bibnamefont {Peres}},\ }\href@noop {} {\emph {\bibinfo {title} {Quantum theory: concepts and methods}}},\ Vol.~\bibinfo {volume} {72}\ (\bibinfo  {publisher} {Springer},\ \bibinfo {year} {1997})\BibitemShut {NoStop}%
\bibitem [{\citenamefont {Terhal}\ and\ \citenamefont {Horodecki}(2000)}]{HorodeckiNumber}%
  \BibitemOpen
  \bibfield  {author} {\bibinfo {author} {\bibfnamefont {B.~M.}\ \bibnamefont {Terhal}}\ and\ \bibinfo {author} {\bibfnamefont {P.}~\bibnamefont {Horodecki}},\ }\bibfield  {title} {\bibinfo {title} {Schmidt number for density matrices},\ }\href {https://doi.org/10.1103/PhysRevA.61.040301} {\bibfield  {journal} {\bibinfo  {journal} {Phys. Rev. A}\ }\textbf {\bibinfo {volume} {61}},\ \bibinfo {pages} {040301} (\bibinfo {year} {2000})}\BibitemShut {NoStop}%
\bibitem [{\citenamefont {Vidal}(2003)}]{Vidal}%
  \BibitemOpen
  \bibfield  {author} {\bibinfo {author} {\bibfnamefont {G.}~\bibnamefont {Vidal}},\ }\bibfield  {title} {\bibinfo {title} {Efficient classical simulation of slightly entangled quantum computations},\ }\href {https://doi.org/10.1103/PhysRevLett.91.147902} {\bibfield  {journal} {\bibinfo  {journal} {Phys. Rev. Lett.}\ }\textbf {\bibinfo {volume} {91}},\ \bibinfo {pages} {147902} (\bibinfo {year} {2003})}\BibitemShut {NoStop}%
\bibitem [{\citenamefont {Van~den Nest}(2013)}]{Nest}%
  \BibitemOpen
  \bibfield  {author} {\bibinfo {author} {\bibfnamefont {M.}~\bibnamefont {Van~den Nest}},\ }\bibfield  {title} {\bibinfo {title} {Universal quantum computation with little entanglement},\ }\href {https://doi.org/10.1103/PhysRevLett.110.060504} {\bibfield  {journal} {\bibinfo  {journal} {Phys. Rev. Lett.}\ }\textbf {\bibinfo {volume} {110}},\ \bibinfo {pages} {060504} (\bibinfo {year} {2013})}\BibitemShut {NoStop}%
\bibitem [{\citenamefont {Huber}\ and\ \citenamefont {de~Vicente}(2013)}]{Huber}%
  \BibitemOpen
  \bibfield  {author} {\bibinfo {author} {\bibfnamefont {M.}~\bibnamefont {Huber}}\ and\ \bibinfo {author} {\bibfnamefont {J.~I.}\ \bibnamefont {de~Vicente}},\ }\bibfield  {title} {\bibinfo {title} {Structure of multidimensional entanglement in multipartite systems},\ }\href {https://doi.org/10.1103/PhysRevLett.110.030501} {\bibfield  {journal} {\bibinfo  {journal} {Phys. Rev. Lett.}\ }\textbf {\bibinfo {volume} {110}},\ \bibinfo {pages} {030501} (\bibinfo {year} {2013})}\BibitemShut {NoStop}%
\bibitem [{\citenamefont {Nielsen}\ and\ \citenamefont {Chuang}(2000)}]{nielsen00}%
  \BibitemOpen
  \bibfield  {author} {\bibinfo {author} {\bibfnamefont {M.~A.}\ \bibnamefont {Nielsen}}\ and\ \bibinfo {author} {\bibfnamefont {I.~L.}\ \bibnamefont {Chuang}},\ }\href@noop {} {\emph {\bibinfo {title} {Quantum Computation and Quantum Information}}}\ (\bibinfo  {publisher} {Cambridge University Press},\ \bibinfo {year} {2000})\BibitemShut {NoStop}%
\bibitem [{\citenamefont {Kumar}(2024)}]{kumarMultipartite}%
  \BibitemOpen
  \bibfield  {author} {\bibinfo {author} {\bibfnamefont {M.}~\bibnamefont {Kumar}},\ }\href {https://arxiv.org/abs/2411.02473} {\bibinfo {title} {Schmidt decomposition of multipartite states}} (\bibinfo {year} {2024}),\ \Eprint {https://arxiv.org/abs/2411.02473} {arXiv:2411.02473 [quant-ph]} \BibitemShut {NoStop}%
\bibitem [{\citenamefont {Schrödinger}(1936)}]{Schrödinger_1936}%
  \BibitemOpen
  \bibfield  {author} {\bibinfo {author} {\bibfnamefont {E.}~\bibnamefont {Schrödinger}},\ }\bibfield  {title} {\bibinfo {title} {Probability relations between separated systems},\ }\href {https://doi.org/10.1017/S0305004100019137} {\bibfield  {journal} {\bibinfo  {journal} {Mathematical Proceedings of the Cambridge Philosophical Society}\ }\textbf {\bibinfo {volume} {32}},\ \bibinfo {pages} {446–452} (\bibinfo {year} {1936})}\BibitemShut {NoStop}%
\bibitem [{\citenamefont {Hughston}\ \emph {et~al.}(1993)\citenamefont {Hughston}, \citenamefont {Jozsa},\ and\ \citenamefont {Wootters}}]{HUGHSTON199314}%
  \BibitemOpen
  \bibfield  {author} {\bibinfo {author} {\bibfnamefont {L.~P.}\ \bibnamefont {Hughston}}, \bibinfo {author} {\bibfnamefont {R.}~\bibnamefont {Jozsa}},\ and\ \bibinfo {author} {\bibfnamefont {W.~K.}\ \bibnamefont {Wootters}},\ }\bibfield  {title} {\bibinfo {title} {A complete classification of quantum ensembles having a given density matrix},\ }\href {https://doi.org/https://doi.org/10.1016/0375-9601(93)90880-9} {\bibfield  {journal} {\bibinfo  {journal} {Physics Letters A}\ }\textbf {\bibinfo {volume} {183}},\ \bibinfo {pages} {14} (\bibinfo {year} {1993})}\BibitemShut {NoStop}%
\bibitem [{\citenamefont {Jaynes}(1957)}]{Jaynes}%
  \BibitemOpen
  \bibfield  {author} {\bibinfo {author} {\bibfnamefont {E.~T.}\ \bibnamefont {Jaynes}},\ }\bibfield  {title} {\bibinfo {title} {Information theory and statistical mechanics. ii},\ }\href {https://doi.org/10.1103/PhysRev.108.171} {\bibfield  {journal} {\bibinfo  {journal} {Phys. Rev.}\ }\textbf {\bibinfo {volume} {108}},\ \bibinfo {pages} {171} (\bibinfo {year} {1957})}\BibitemShut {NoStop}%
\bibitem [{\citenamefont {Hadjisavvas}(1981)}]{Hadjisavvas}%
  \BibitemOpen
  \bibfield  {author} {\bibinfo {author} {\bibfnamefont {N.}~\bibnamefont {Hadjisavvas}},\ }\bibfield  {title} {\bibinfo {title} {Properties of mixtures on non-orthogonal states},\ }\href {https://doi.org/https://doi.org/10.1007/BF00401481} {\bibfield  {journal} {\bibinfo  {journal} {Lett Math Phys 5}\ ,\ \bibinfo {pages} {327–332}} (\bibinfo {year} {1981})}\BibitemShut {NoStop}%
\bibitem [{\citenamefont {Gisin}(1989)}]{Gisin}%
  \BibitemOpen
  \bibfield  {author} {\bibinfo {author} {\bibfnamefont {N.}~\bibnamefont {Gisin}},\ }\bibfield  {title} {\bibinfo {title} {Stochastic quantum dynamics and relativity},\ }\href@noop {} {\bibfield  {journal} {\bibinfo  {journal} {Helvetica Physica Acta 62}\ ,\ \bibinfo {pages} {363–371}} (\bibinfo {year} {1989})}\BibitemShut {NoStop}%
\end{thebibliography}%

\end{document}